\documentclass[11pt]{article}

\usepackage{amsmath}
\usepackage{bm}
\usepackage{amssymb}
\usepackage{amstext}
\usepackage{latexsym}
\usepackage{color}
\usepackage{graphicx}
\usepackage{epsfig}
\usepackage{subfigure}
\usepackage{rotating}
\usepackage{curves}
\usepackage{epic}
\usepackage{amsthm}
\allowdisplaybreaks[1]

\usepackage{amsfonts}
\usepackage[T1]{fontenc}
\usepackage[latin1]{inputenc}
\usepackage{authblk}
\usepackage{tikz}
\usetikzlibrary{shapes,arrows}
\usepackage{graphics}
\usepackage{verbatim}
\usepackage{color}
\usepackage{hyperref}
\usepackage{mathtools}

\newcommand{\be}{\begin{equation}}
\newcommand{\en}{\end{equation}}

\newtheorem{thm}{Theorem}
\newtheorem{cor}[thm]{Corollary}
\newtheorem{proposition}[thm]{Proposition}

\newtheorem{defi}{Definition}[section]
\newtheorem{lem}[defi]{Lemma}

\newtheorem{Theo}{Theorem}[section]

\newtheorem{remark}[Theo]{Remark}

\newcommand{\bedefin}{\begin{defi}}
\newcommand{\findefi}{\end{defi} \medskip}

\newcommand{\betheo}{\begin{theorem}$\!\!${\bf \,\,\,}}
\newcommand{\entheo}{\end{theorem}}
\newcommand{\enth}{\end{theorem}}

\newcommand{\becor}{\begin{cor}$\!\!${\bf .}}
\newcommand{\encor}{\end{cor}}

\newcommand{\belem}{\begin{lem}$\!\!${\bf }}
\newcommand{\enlem}{\end{lem}}

\newcommand{\bea}{\begin{eqnarray}}
\newcommand{\ena}{\end{eqnarray}}

\newcommand{\beano}{\begin{eqnarray*}}
\newcommand{\enano}{\end{eqnarray*}}

\newcommand{\bee}{\begin{enumerate}}
\newcommand{\ene}{\end{enumerate}}

\newcommand{\bei}{\begin{itemize}}
\newcommand{\eni}{\end{itemize}}

\newcommand{\betab}{\begin{tabular}}
\newcommand{\entab}{\end{tabular}}

\newcommand{\bd}{\begin{displaymath}}

\newcommand{\ba}{\mathbf a}
\newcommand{\bb}{\mathbf b}

\newcommand{\bp}{\mathbf p}
\newcommand{\bq}{\mathbf q}

\newcommand{\bA}{\mathbf A}

\newcommand{\g}{G_{\hbox{\tiny{NC}}}}
\newcommand{\gd}{\hat{G}_{\hbox{\tiny{NC}}}}
\newcommand{\G}{\mathfrak{g}_{\hbox{\tiny{NC}}}}

\DeclareMathOperator{\tr}{Tr}

\catcode `\@=11 \@addtoreset{equation}{section}

\catcode `\@=12

\begin{document}

\title{Wigner functions for gauge equivalence classes of unitary irreducible representations of noncommutative quantum mechanics}
\author[1]{S. Hasibul Hassan Chowdhury\thanks{shhchowdhury@gmail.com}}
\author[2]{Hishamuddin Zainuddin\thanks{hisham@upm.edu.my}}
\affil[1,2]{Laboratory of Computational Sciences and Mathematical Physics, Institute for Mathematical Research, Universiti Putra Malaysia, 43400 UPM Serdang, Selangor, Malaysia}
\affil[2]{Malaysia-Italy Centre of Excellence for Mathematical Sciences, Universiti Putra Malaysia, Malaysia}
\date{}

\maketitle

\begin{abstract}
While Wigner functions forming phase space representation of quantum states is a well-known fact, their construction for noncommutative quantum mechanics (NCQM) remains relatively lesser known, in particular with respect to gauge dependencies. This paper deals with the construction of Wigner functions of NCQM for a system of 2-degrees of freedom using 2-parameter families of gauge equivalence classes of unitary irreducible representations (UIRs) of the Lie group $\g$ which has been identified as the kinematical symmetry group of NCQM in an earlier paper. This general construction of Wigner functions for NCQM, in turn, yields the special cases of Landau and symmetric gauges of NCQM.
\end{abstract} 

\section{Introduction}
\label{sec:intro}
Noncommutative quantum mechanics (NCQM) has been studied extensively as a possible modification of quantum mechanics at the Planck scale for which the position observables are noncommuting (hence dimension of configuration space is at least two). The motivation of the introduction of NCQM directs us back to the seminal work of Snyder \cite{Sny} where he studied the quantized structure of space-time in a Lorentz invariant fashion. Such a model of space-time in scales as small as Planck length is also proposed, among others by Doplicher et al. \cite{Doplicheretal} in order to avoid creation of microscopic black holes to the effect of losing the operational meaning of localization in space-time. NCQM can also be realized as nonrelativistic approximation of noncommutative field theory (see \cite{Hoetal}).  A noncommutative quantum field theory (NCQFT) is the one where the fields are functions of space-time coordinates with spatial coordinates failing to commute with each other. A multitude of articles appearing in the literature (see, for example, \cite{Douglasetal,Connesetal,Seibergetal,Chuetal1,Chuetal2,Ardalanetal,Schomerus}) focus on how NCQFT is related to the studies of {\em noncommutative geometry} and {\em string theory}. On the other hand, how noncommutative geometry can be studied independently from the unitary irreducible representations of NCQM are all explored in \cite{plethora}. 

By instituting noncommutativity of the position observables, one can also take it to be in the context of restoring the symmetry between phase space coordinates when treating a quantum system of particles in a magnetic field where the momenta are noncommuting. The underlying general kinematical symmetry group $G_{NC}$ of NCQM in its smallest dimension was understood only relatively recently as the triply extended groups of translations in $\mathbb{R}^4$ \cite{ncqmjpa}. One of us \cite{plethora} has explored exhaustively the different inequivalent representations of $G_{NC}$ and showed that they fall into three classes of 4-dimensional, 2-dimensional and 0-dimensional coadjoint orbits. It was observed that within the 4-dimensional coadjoint orbit, there exists families of representations labelled by three parameters of which one case (with two parameters tending to zero) corresponds to standard quantum mechanics. In what follows, we will limit our interest to the case of where all three parameters are nonzero and a certain constraint between these parameters hold.

Having restored the symmetry of the phase-space observables through their non-commutativities, the plausible natural setting for the full-fledged NCQM is the Wigner function representations. Wigner function representations \cite{Wigner} are known to be the quantum-mechanical analogue of the phase space distribution of classical statistical mechanics \cite{Moyal} that includes quantum mechanical corrections and hence serve many computational purposes in quantum settings. It is also known that the Wigner function can be related to variation of information measures and entropies (see \cite{OlivaresEtAl}). General features of time independent Wigner functions are studied in \cite{Curtrightetal}. The quasi-probabilistic nature of Wigner function may pose problems in interpretation but one can smoothen the function over phase space area elements to form the more well-behaved Husimi distribution (see \cite{HWLee}). Another problem discussed in constructing Wigner functions is its gauge-dependency when treating noncommutativity of momenta operators in the Landau system of (charged) particle in magnetic field. For the full NCQM, it is expected that more gauge dependency will appear with further noncommuting positions. The gauge dependencies of the Wigner functions are rendered harmless in the group-theoretic formulation since there will be unitary operators transforming from one equivalent gauged representation to another. In the case of NCQM, the NC Wigner functions have been calculated for the cases of the Landau and symmetric gauges in \cite{wigfuncpaper}. It is thus of interest to us to explore the NC Wigner functions in the general gauge setting and show explicitly a two-gauge parameter dependence, allowing access to more computational possibilities.

\section{On the Lie group $\g$, the kinematical symmetry group of NCQM}
\label{sec:Preli}
The group $\g$, being a 7-dimensional connected, simply connected nilpotent Lie group was first introduced in \cite{ncqmjmp}, the role of which in 2-dimensional NCQM runs parallel to that of the 5-dimensional Weyl-Heisenberg group for quantum mechanics in 2-dimensions. To have a brief comparative review of these two nilpotent Lie groups consult section II of \cite{plethora}. The group $\g$ is a seven parameter, real Lie group. We shall write a general element of the group as
\begin{equation}
  g = (\theta, \phi, \psi, \bq , \bp), \qquad \theta, \phi, \psi \in \mathbb R, \quad \bq = (q_1, q_2) \in \mathbb R^2, \quad \bp = (p_1, p_2) \in \mathbb R^2 ,
\label{group-elem}
\end{equation}
with the group multiplication given by
\begin{eqnarray}
\label{grplawtriplyextendedgrp}
\lefteqn{(\theta,\phi,\psi,\bq,\bp)(\theta^{\prime},\phi^{\prime},\psi^{\prime},
\bq^{\prime},\bp^{\prime})}\nonumber\\
&&=(\theta+\theta^{\prime}+\frac{\alpha}{2}[\langle\bq \cdot \bp^{\prime}\rangle-\langle
\bp \cdot\bq^{\prime}\rangle],\phi+\phi^{\prime}+\frac{\beta}{2}[\bp\wedge
\bp^{\prime}],\psi+\psi^{\prime}+\frac{\gamma}{2}[\bq\wedge\bq^{\prime}],
\nonumber\\
&&\;\;\;\; \bq+\bq^{\prime},\bp+\bp^{\prime}).
\end{eqnarray}
where, for two 2-vectors, $\ba = (a_1, a_2), \bb = (b_1, b_2 ), \; \ba \cdot \bb = a_1 b_1 + a_2 b_2$ and $\ba\wedge \bb = a_1 b_2 - a_2 b_1$. The quantities $\alpha, \beta$ and $\gamma$ are dimensional constants. The three parameters $\theta, \phi, \psi$ constitute the centre of the group, while the $\bq$ and $\bp$, are the parameters of $\mathbb R^4$, the three-fold central extension of which leads to $\g$. Denoting the centre by $\mathcal Z$, we see that $\g/\mathcal Z \simeq \mathbb R^4$. Note that if we denote the dimension of the position coordinate by $[q]$ and that of the momentum coordinate by $[p]$, then we immediately see that in order to have $\theta$, $\phi$ and $\psi$ to be all dimensionless, we must have $[\alpha]=\left[\frac{1}{pq}\right]$, $[\beta]=\left[\frac{1}{p^2}\right]$ and $[\gamma]=\left[\frac{1}{q^2}\right]$.

We know from \cite{ncqmjpa} that $\g$ admits coadjoint orbits of dimension 4, 2 and 0 which are completely determined by the real triple $(\rho,\sigma,\tau)$. It was also found there that the unitary dual of $\g$ is in 1-1 correspondence with its various coadjoint orbits and hence s labelled by the same triple $(\rho,\sigma,\tau)$. Respecting the notations used in \cite{ncqmjpa}, a generic coadjoint orbit of $\g$ will be denoted by $\mathcal{O}^{\rho,\sigma,\tau}$. In this paper, the 4 dimensional coadjoint orbits with nonzero $\rho$, $\sigma$ and $\tau$ satisfying $\rho^{2}\alpha^{2}-\sigma\beta\gamma\tau\neq 0$ will only concern us. Such orbits will be denoted as $\mathcal{O}^{\rho,\sigma,\tau}_{4}$.

\section{Gauge equivalence classes of unitary irreducible representations of noncommutative quantum mechanics}
\label{sec:gauge-class}
In \cite{plethora}, 2-parameter $(l,m)$ gauge equivalence classes of UIRs of $\g$ were computed. It was also shown there that they, in turn, give rise to vector potentials labeled by $m$. The 2-parameter family of unitarily equivalent irreducible representations (UIRs) of $\g$ was due to a fixed coadjoint orbit determined by $\rho=\sigma=\tau=1$ where the unitary dual of $\g$ is labeled by the triple $(\rho,\sigma,\tau)$. When one fixes this triple, by varying $(l,m)$, one can choose a representative of the underlying equivalence class of unitary irreducible representations of $\g$. The set of UIRs $U^{\rho,\sigma,\tau}_{l,m}$ for $\rho\neq 0$, $\sigma\neq 0$, $\tau\neq 0$ and $\rho^{2}-\gamma\beta\sigma\tau\neq 0$ with fixed $m\in\mathbb{R}$ and $l\in\mathbb{R}\smallsetminus\{\frac{\rho^{2}\alpha^{2}}{\gamma\beta\sigma\tau}\}$ will be denoted by $\mathfrak{N}$. These UIRs read as follows
\begin{eqnarray}\label{eq:generic-rep}
\lefteqn{(U^{\rho,\sigma,\tau}_{l,m}(\theta,\phi,\psi,\vec{q},\vec{p})f)(r_{1},r_{2})}\nonumber\\
&&=e^{-i\rho\theta-i\sigma\phi-i\tau\psi}e^{i\rho\alpha p_{1}r_{1}+i\rho\alpha p_{2}r_{2}+\frac{i\rho^{2}\alpha^{2}\tau\gamma(1-l)}{\tau\gamma\sigma\beta l-\rho^{2}\alpha^{2}}q_{1}r_{2}+il\tau\gamma q_{2}r_{1}+i\left[\frac{\rho\alpha}{2}+\frac{\rho\alpha\tau\gamma\sigma\beta m(1-l)}{\tau\gamma\sigma\beta l-\rho^{2}\alpha^{2}}\right]p_{1}q_{1}}\nonumber\\
&&\times e^{i\left[\frac{\rho\alpha}{2}-\frac{l\tau\gamma\sigma\beta(1-m)}{\rho\alpha}\right]p_{2}q_{2}+i(m-\frac{1}{2})\sigma\beta p_{1}p_{2}+i\left[\frac{\tau\gamma}{2}-\frac{\tau\gamma(1-l)(\tau\gamma\sigma\beta l-\tau\gamma\sigma\beta lm-\rho^{2}\alpha^{2})}{\tau\gamma\sigma\beta l-\rho^{2}\alpha^{2}}\right]q_{1}q_{2}}\nonumber\\
&&\scriptstyle\times f\left(r_{1}-\frac{(1-m)\sigma\beta}{\rho\alpha}p_{2}+\frac{\tau\gamma\sigma\beta(l+m-lm)-\rho^{2}\alpha^{2}}{\tau\gamma\sigma\beta l-\rho^{2}\alpha^{2}}q_{1},r_{2}+\frac{m\sigma\beta}{\rho\alpha}p_{1}-\frac{\tau\gamma\sigma\beta l(1-m)-\rho^{2}\alpha^{2}}{\rho^{2}\alpha^{2}}q_{2}\right),
\end{eqnarray}
where $f\in L^{2}(\mathbb{R}^{2},dr_{1}dr_{2})$.

The corresponding self adjoint representation of $\G$ acting on the smooth vectors of $L^{2}(\mathbb{R}^{2},dr_{1}dr_{2})$ is given by
\begin{equation}\label{gauge-equiv-reps-algbr}
\begin{split}
&\hat{Q}^{m}_{1}=r_{1}-m\frac{i\sigma\beta}{\rho^{2}\alpha^{2}}\frac{\partial}{\partial r_{2}},\\
&\hat{Q}^{m}_{2}=r_{2}+(1-m)\frac{i\sigma\beta}{\rho^{2}\alpha^{2}}\frac{\partial}{\partial r_{1}},\\
&\hat{P}^{l,m}_{1}=\frac{\tau\gamma\rho\alpha(1-l)}{\tau\gamma\sigma\beta l-\rho^{2}\alpha^{2}}r_{2}-\frac{i}{\rho\alpha}\left[\frac{\tau\gamma\sigma\beta(l+m-lm)-\rho^{2}\alpha^{2}}{\tau\gamma\sigma\beta l-\rho^{2}\alpha^{2}}\right]\frac{\partial}{\partial r_{1}},\\
&\hat{P}^{l,m}_{2}=\frac{l\tau\gamma}{\rho\alpha}r_{1}+i\left[\frac{\tau\gamma\sigma\beta l(1-m)-\rho^{2}\alpha^{2}}{\rho^{3}\alpha^{3}}\right]\frac{\partial}{\partial r_{2}}.
\end{split}
\end{equation}

It is interesting note that these operators are slightly different from the usual notion of Wigner operators\cite{Kubo} for which the polarized form of operators are in place as expected in a quantization scheme. Rearranging the terms in the last two equations of (\ref{gauge-equiv-reps-algbr}) using the first two then leads us to
\begin{equation}\label{gauge-equiv-reps-algbr-rarrngd}
\begin{split}
&\hat{Q}^{m}_{1}=r_{1}-m\frac{i\sigma\beta}{\rho^{2}\alpha^{2}}\frac{\partial}{\partial r_{2}},\\
&\hat{Q}^{m}_{2}=r_{2}+(1-m)\frac{i\sigma\beta}{\rho^{2}\alpha^{2}}\frac{\partial}{\partial r_{1}},\\
&\hat{P}^{l,m}_{1}=\frac{\tau\gamma\rho\alpha(1-l)}{\tau\gamma\sigma\beta l-\rho^{2}\alpha^{2}}\hat{Q}^{m}_{2}-\frac{i}{\rho\alpha}\left[\frac{\tau\gamma\sigma\beta(1-l)}{\tau\gamma\sigma\beta l-\rho^{2}\alpha^{2}}+1\right]\frac{\partial}{\partial r_{1}},\\
&\hat{P}^{l,m}_{2}=\frac{l\tau\gamma}{\rho\alpha}\hat{Q}^{m}_{1}-\frac{i}{\rho\alpha}\left(1-\frac{l\tau\gamma\sigma\beta}{\rho^{2}\alpha^{2}}\right)\frac{\partial}{\partial r_{2}}.
\end{split}
\end{equation}

One is then motivated by (\ref{gauge-equiv-reps-algbr-rarrngd}) to define the underlying vector potential $\bA^{\rho,\sigma,\tau}\equiv(A^{\rho,\sigma,\tau}_{1},A^{\rho,\sigma,\tau}_{2})$ for the 2-dimensional system of NCQM as

\begin{equation} \label{vector-pot-def}
\bA^{\rho,\sigma,\tau}\equiv\left(-\frac{\tau\gamma\rho\alpha(1-l)}{\tau\gamma\sigma\beta l-\rho^{2}\alpha^{2}}\hat{Q}^{m}_{2},-\frac{l\tau\gamma}{\rho\alpha}\hat{Q}^{m}_{1}\right). 
\end{equation}

Now, if the constant magnetic field  applied perpendicular to the underlying 2-dimensional system is denoted by $B$, then from the discussion conducted at p. 14 in \cite{ncqmjpa}, the magnetic field reads off as $B=-\frac{\tau\gamma}{\rho\alpha}$. It becomes evident then that, as in a quantum mechanical system in the presence of a vertical constant magnetic field $B$ (without spatial noncommutativity taken into consideration), the relation between the magnetic field and the underlying vector potential, i.e. $\partial_{1}A^{\rho,\sigma,\tau}_{2}-\partial_{2}A^{\rho,\sigma,\tau}_{1}=B$, no longer holds, where $A^{\rho,\sigma,\tau}_{i}$'s, for $i=1,2$, are the 2-components of the vector potential $\bA^{\rho,\sigma,\tau}$. The limiting expression of $\bA^{\rho,\sigma,\tau}$, as $\sigma\rightarrow 0$, denoted by $\bA^{\rho,\tau}\equiv(A^{\rho,\tau}_{1},A^{\rho,\tau}_{2})$ then is easily found to restore the relation $\partial_{1}A^{\rho,\tau}_{2}-\partial_{2}A^{\rho,\tau}_{1}=B$.

The well-known {\em symmetric gauge} representation of NCQM corresponds to the choice $l=\frac{\rho\alpha(\rho\alpha-\sqrt{\rho^{2}\alpha^{2}-\gamma\beta\sigma\tau})}{\gamma\beta\sigma\tau}:=l_{s}$ and $m=\frac{1}{2}$ in (\ref{gauge-equiv-reps-algbr}) as is given below:
\begin{equation}\label{sym-gauge-rep-algbr}
\begin{split}
&\hat{Q}^{\frac{1}{2}}_{1}=r_{1}-\frac{i\sigma\beta}{2\rho^{2}\alpha^{2}}\frac{\partial}{\partial r_{2}},\\
&\hat{Q}^{\frac{1}{2}}_{2}=r_{2}+\frac{i\sigma\beta}{2\rho^{2}\alpha^{2}}\frac{\partial}{\partial r_{1}},\\
&\hat{P}^{l_{s},\frac{1}{2}}_{1}=\frac{(\sqrt{\rho^{2}\alpha^{2}-\gamma\beta\sigma\tau}-\rho\alpha)}{\sigma\beta}r_{2}-\frac{i}{2\rho^{2}\alpha^{2}}(\rho\alpha+\sqrt{\rho^{2}\alpha^{2}-\gamma\beta\sigma\tau})\frac{\partial}{\partial r_{1}},\\
&\hat{P}^{l_{s},\frac{1}{2}}_{2}=\frac{(\rho\alpha-\sqrt{\rho^{2}\alpha^{2}-\gamma\beta\sigma\tau})}{\sigma\beta}r_{1}-\frac{i}{2\rho^{2}\alpha^{2}}(\rho\alpha+\sqrt{\rho^{2}\alpha^{2}-\gamma\beta\sigma\tau})\frac{\partial}{\partial r_{2}}.
\end{split}
\end{equation}
The corresponding vector potential $\bA^{\rho,\sigma,\tau}_{\hbox{\tiny{sym}}}$, from (\ref{vector-pot-def}), then reads off as
\begin{equation} \label{vector-pot-symmetric-gauge}
\bA^{\rho,\sigma,\tau}_{\hbox{\tiny{sym}}}\equiv\left(\frac{(\rho\alpha-\sqrt{\rho^{2}\alpha^{2}-\gamma\beta\sigma\tau})}{\sigma\beta}\hat{Q}^{\frac{1}{2}}_{2},\frac{(\sqrt{\rho^{2}\alpha^{2}-\gamma\beta\sigma\tau}-\rho\alpha)}{\sigma\beta}\hat{Q}^{\frac{1}{2}}_{1}\right). 
\end{equation}
Denote the 2-components of the vector potential $\bA^{\rho,\sigma,\tau}_{\hbox{\tiny{sym}}}$ by $A^{\rho,\sigma\tau}_{i,\hbox{\tiny{sym}}}$ with $i=1,2$ and observe that the following holds
\begin{equation}\label{def-position-nc-magnetic-field}
\partial_{1}A^{\rho,\sigma,\tau}_{2,\hbox{\tiny{sym}}}-\partial_{2}A^{\rho,\sigma,\tau}_{1,\hbox{\tiny{sym}}}=\frac{2\hbar}{\vartheta}\left(\sqrt{1-\frac{B\vartheta}{\hbar}}-1\right):=\bar{B},
\end{equation}
where we chose $B=-\frac{\tau\gamma}{\rho\alpha}$, $\vartheta=-\frac{\sigma\beta}{\rho^{2}\alpha^{2}}$ and $\hbar=\frac{1}{\rho\alpha}$ as in p. 14 of \cite{ncqmjpa}. Compare (\ref{def-position-nc-magnetic-field}) with (43) at p. 14 of \cite{DelducEtAl2008}.

Also, of considerable interest in Physics literature, is the {\em Landau gauge} representation of NCQM that corresponds to $l=1$, $m=0$ in (\ref{gauge-equiv-reps-algbr}). The self adjoint representation of $\G$ in the Landau gauge representation can be written as
\begin{equation}\label{Landau-gauge-rep-algbr}
\begin{split}
&\hat{Q}^{0}_{1}=r_{1},\\
&\hat{Q}^{0}_{2}=r_{2}+\frac{i\sigma\beta}{\rho^{2}\alpha^{2}}\frac{\partial}{\partial r_{1}},\\
&\hat{P}^{1,0}_{1}=-\frac{i}{\rho\alpha}\frac{\partial}{\partial r_{1}},\\
&\hat{P}^{1,0}_{2}=\frac{\tau\gamma}{\rho\alpha}r_{1}+\frac{i(\tau\gamma\sigma\beta-\rho^{2}\alpha^{2})}{\rho^{3}\alpha^{3}}\frac{\partial}{\partial r_{2}}.
\end{split}
\end{equation}

Substituting $l=1$, $m=0$ in (\ref{vector-pot-def}) then yields the Landau gauge vector potential 
\begin{equation} \label{vector-pot-Landau-gauge}
\bA^{\rho,\sigma,\tau}_{\hbox{\tiny{Landau}}}\equiv(0,-\frac{\tau\gamma}{\rho\alpha}\hat{Q}^{0}_{1})=(0,Br_{1}),
\end{equation}
recovering $\partial_{1}A^{\rho,\sigma,\tau}_{2,\hbox{\tiny{Landau}}}-\partial_{2}A^{\rho,\sigma,\tau}_{1,\hbox{\tiny{Landau}}}=B$ where $A^{\rho,\sigma,\tau}_{i,\hbox{\tiny{Landau}}}$'s with $i=1,2$, being the components of the vector potential $\bA^{\rho,\sigma,\tau}_{\hbox{\tiny{Landau}}}$. This is again in agreement with what Delduc et al. found in \cite{DelducEtAl2008} (see p. 15).

\section{Wigner functions for equivalence classes of UIRs of $\g$}
\label{sec:Wig-func}
In this section, we will be dealing with he construction of Wigner function associated with a 2-dimensional system of NCQM using the general method developed in \cite{plancherel}. The UIRs of $\g$ that we will be focussing on for this purpose are given by (\ref{eq:generic-rep}). For this purpose, we first need to compute the Plancherel measure that the unitary dual of $\g$ is equipped with. We will then employ this measure to compute the NCQM Wigner function (see (\ref{wig-fcn-def})) for these gauge equivalence classes of UIRs of $\g$.

\subsection{Plancherel measure associated with the sector $\mathfrak{N}$ of $\gd$}
\label{subsec:Plancherel}

Let us denote by $\mathfrak{H}^{\rho,\sigma,\tau}$, a copy of $L^{2}(\mathbb{R}^{2},dr_{1}dr_{2})$ and consider the Hilbert bundle based on the unitary dual $\gd$, a generic point of which has coordinates $(\rho,\sigma,\tau)$ and the fiber at each $(\rho,\sigma,\tau)$ being the Hilbert space of Hilbert-Schmidt operators on $\mathfrak{H}^{\rho,\sigma,\tau}$ denoted by $\mathcal{B}_{2}(\rho,\sigma,\tau)$. Now consider the measurable fields of operators $(\rho,\sigma,\tau)\longmapsto\mathcal{A}(\rho,\sigma,\tau)$ with each $\mathcal{A}(\rho,\sigma,\tau)\in\mathcal{B}_{2}(\rho,\sigma,\tau)$. The space of such measurable fields of operators is endowed with an inner product structure and is called a {\em direct integral Hilbert space}. Given two such fields $(\rho,\sigma,\tau)\longmapsto\mathcal{A}^{1}(\rho,\sigma,\tau)$ and $(\rho,\sigma,\tau)\longmapsto\mathcal{A}^{2}(\rho,\sigma,\tau)$, the inner product between them is given by
\begin{equation}\label{inner-prod-dir-Hilbertspace}
\langle\mathcal{A}^{1}\vert\mathcal{A}^{2}\rangle_{\mathcal{B}_{2}^{\oplus}}=\int_{\gd}\tr(\mathcal{A}^{1}(\rho,\sigma,\tau)^{*}\mathcal{A}^{2}(\rho,\sigma,\tau))d\nu_{\g}(\rho,\sigma,\tau),
\end{equation}
where $d\nu_{\g}$ is the well-known Plancherel measure that the unitary dual $\g$ can be endowed with and we denote the direct integral Hilbert space of measurable fields of Hilbert-Schmidt operators by $\mathcal{B}_{2}^{\oplus}=\int_{\gd}^{\oplus}\mathcal{B}_{2}(\rho,\sigma,\tau)d\nu_{\g}(\rho,\sigma,\tau)$ rendering the fact that the inner product of such Hilbert space is given by (\ref{inner-prod-dir-Hilbertspace}).

Now the Plancherel measure for $\gd$ can be computed using a general orthogonality relation (see, for example, \cite{plancherel}) given by
\begin{eqnarray}\label{eq:ortho cond}
\lefteqn{\int_{\g}\left[\int_{\gd}\tr(U^{\rho,\sigma,\tau}_{l,m}(g)^{*}A^{1}(\rho,\sigma,\tau)
C_{\rho,\sigma,\tau}^{-1})d\nu_{\g}(\rho,\sigma,\tau)\right.}\nonumber\\
&\left.\times\int_{\gd}\tr(U^{\rho^{\prime},\sigma^{\prime},\tau^{\prime}}_{l,m}(g)^{*}A^{2}
(\rho^{\prime},\sigma^{\prime},\tau^{\prime})C^{-1}_{\rho^{\prime},\sigma^{\prime},
\tau^{\prime}})d\nu_{\g}(\rho^{\prime},\sigma^{\prime},\tau^{\prime})\right]d\mu(g)=\langle A^{1}|A^{2}\rangle_{\mathcal{B}_{2}^{\oplus}}\nonumber\\
\end{eqnarray}
The group $\g$ is unimodular and $d\mu$ is the Haar measure on it; $C_{\rho,\sigma,\tau}$ is the {\em Duflo-Moore operator} \cite{Du-Mo}, for the representation $U^{\rho,\sigma,\tau}_{l,m}(g)$ given by (\ref{eq:generic-rep}).

Now for the measurable vector fields $(\rho,\sigma,\tau)\longmapsto\lambda_{\rho,\sigma,\tau}$ and $(\rho,\sigma,\tau)\longmapsto\chi_{\rho,\sigma,\tau}$ with the vectors $\lambda_{\rho,\sigma,\tau},\chi_{\rho,\sigma,\tau}\in L^{2}(\mathbb{R}^{2},dr_{1}dr_{2})$, one can choose the operator fields $\mathcal{A}^{1}=\mathcal{A}^{2}$ in (\ref{eq:ortho cond}) using the rank 1-operators $\mathcal{A}^{1}(\rho,\sigma,\tau)=\mathcal{A}^{2}(\rho,\sigma,\tau)=\vert\chi_{\rho,\sigma,\tau}\rangle\langle\lambda_{\rho,\sigma,\tau}\vert$. Then writing the Plancherel measure $d\nu_{\g}$ with the help of a suitable density $\kappa(\rho,\sigma,\tau)$ as $d\nu_{\g}(\rho,\sigma,\tau)=\kappa(\rho,\sigma,\tau)d\rho d\sigma d\tau$, one computes the right side of (\ref{eq:ortho cond}) restricted to the sector $\mathfrak{N}$ of $\gd$ using the inner product defined by (\ref{inner-prod-dir-Hilbertspace}).
\begin{eqnarray}\label{eq:ortho-con-presnt-secnario}
\langle A^{1}|A^{1}\rangle_{\mathcal{B}_{2}^{\oplus}}&=&\int_{\mathfrak N}\tr(A^{1}(\rho,\sigma,\tau)^{*}A^{1}(\rho,\sigma,\tau))\kappa(\rho,\sigma,\tau)d\rho d\sigma d\tau\nonumber\\
&=&\int_{\mathfrak N}\tr(|\lambda_{\rho,\sigma,\tau}\rangle\langle\chi_{\rho,\sigma,\tau}|\chi_{\rho,\sigma,\tau}\rangle\langle\lambda_{\rho,\sigma,\tau}|)\kappa(\rho,\sigma,\tau)d\rho d\sigma d\tau\nonumber\\
&=& \int_{\mathfrak N}\|\chi_{\rho,\sigma,\tau}\|^{2}\|\lambda_{\rho,\sigma,\tau}\|^{2}\kappa(\rho,\sigma,\tau)d\rho d\sigma d\tau.
\end{eqnarray}

In \cite{wigfuncpaper}, the Plancherel measure of the unitary dual $\gd$ restricted to the sector with $\rho\neq 0$, $\sigma\neq 0$, $\tau\neq 0$ and $\rho^{2}\alpha^{2}-\gamma\beta\sigma\tau\neq 0$ has already been computed (see proposition 1, p. 4). There, the unitary irreducible representation of $\g$, i.e. the representative chosen for each distinct value of the triple $(\rho,\sigma,\tau)$ was the one associated with $l=1$ and $m=1$ in (\ref{eq:generic-rep}). But the Plancherel measure of the unitary dual $\gd$ restricted to the sector $\rho\neq 0$, $\sigma\neq 0$, $\tau\neq 0$ and $\rho^{2}\alpha^{2}-\gamma\beta\sigma\tau\neq 0$ should be independent of the choice of the representative for each $(\rho,\sigma,\tau)$ as has been verified in the following proposition:
\begin{proposition}\label{prop:Plancherel}
If one considers the unitary dual $\gd$ restricted to the sector $\rho\neq 0$, $\sigma\neq 0$, $\tau\neq 0$ and $\rho^{2}\alpha^{2}-\gamma\beta\sigma\tau\neq 0$ where each equivalence class of UIRs of $\g$ is represented by the representation (\ref{eq:generic-rep}) for a fixed ordered pair $(l,m)$, then the Plancherel measure of $\gd$ restricted to such family of UIRs of $\g$ is given by
\begin{equation}\label{eq:Plancherel-prop}
d\nu_{\g}(\rho,\sigma,\tau)=\frac{\vert\rho^{2}\alpha^{2}-\gamma\beta\sigma\tau\vert}{\alpha^{2}}d\rho d\sigma d\tau,
\end{equation}
and the corresponding Duflo-Moore operator reads
\begin{equation}\label{eq:duflo-moore}
C_{\rho,\sigma,\tau}=(2\pi)^{\frac{5}{2}}\mathbb{I},
\end{equation}
where $\mathbb{I}$ is the identity operator on $L^{2}(\mathbb{R}^{2},dr_{1}dr_{2})$.
\end{proposition}

\begin{proof}
Given the fact that $\g$ is unimodular, the underlying Duflo-Moore operator reads $C_{\rho,\sigma,\tau}=N\mathbb{I}$, where $N$ is a real number and $\mathbb{I}$ is the identity operator on $L^{2}(\mathbb{R}^{2},dr_{1}dr_{2})$. Now the left side of (\ref{eq:ortho cond}) can be read off as
\begin{eqnarray}\label{eq:planchrl-dervtn}
\lefteqn{\frac{1}{N^2}\int_{\mathbb{R}^7}\left[\int_{\mathbb{R}^{*}\times\mathbb{R}^{*}\times\mathbb{R}^{*}}\overline{\langle\chi_{\rho,\sigma,\tau}|U^{\rho,\sigma,\tau}_{l,m}(\theta,\phi,\psi,\bq,\bp)\lambda_{\rho,\sigma,\tau}\rangle}\kappa(\rho,\sigma,\tau)d\rho d\sigma d\tau\right.}\nonumber\\
&&\scriptstyle\left.\times\int_{\mathbb{R}^{*}\times\mathbb{R}^{*}\times\mathbb{R}^{*}}\langle\chi_{\rho^{\prime},\sigma^{\prime},\tau^{\prime}}|U^{\rho^{\prime},\sigma^{\prime},\tau^{\prime}}_{l,m}(\theta,\phi,\psi,\bq,\bp)\lambda_{\rho^{\prime},\sigma^{\prime},\tau^{\prime}}\rangle\kappa(\rho^{\prime},\sigma^{\prime},\tau^{\prime})d\rho^{\prime}d\sigma^{\prime}d\tau^{\prime}\right] d\theta\; d\phi\; d\psi\; d\bq\; d\bp\nonumber\\
=\lefteqn{\textstyle\frac{1}{N^2}\int_{(\rho,\sigma,\tau)}\int_{(\rho^{\prime},\sigma^{\prime},\tau^{\prime})}\left[\int_{\mathbb{R}^{7}}\left\{\int_{(r_{1},r_{2})\in\mathbb{R}^{2}}\int_{(r_{1}^{\prime},r_{2}^{\prime})\in\mathbb{R}^{2}}e^{i(\rho-\rho^{\prime})\theta+i(\sigma-\sigma^{\prime})\phi+i(\tau-\tau^{\prime})\psi}\right.\right.}\nonumber\\
&&\left.\left.\times e^{-i\alpha p_{1}(\rho r_{1}-\rho^{\prime}r_{1}^{\prime})-i\alpha p_{2}(\rho r_{2}-\rho^{\prime}r_{2}^{\prime})-i\alpha^{2}\gamma(1-l)q_{1}\left(\frac{\rho^{2}\tau r_{2}}{\tau\gamma\sigma\beta l-\rho^{2}\alpha^{2}}-\frac{\rho^{\prime 2}\tau^{\prime}r^{\prime}_{2}}{\tau^{\prime}\gamma\sigma^{\prime}\beta l-\rho^{\prime 2}\alpha^{2}}\right)}\right.\right.\nonumber\\
&&\left.\left.\times e^{-il\gamma q_{2}(\tau r_{1}-\tau^{\prime}r_{1}^{\prime})-i\left[\frac{(\rho-\rho^{\prime})\alpha}{2}+\alpha\beta\gamma m(1-l)\left(\frac{\rho\sigma\tau}{\tau\gamma\sigma\beta l-\rho^{2}\alpha^{2}}-\frac{\rho^{\prime}\sigma^{\prime}\tau^{\prime}}{\tau^{\prime}\gamma\sigma^{\prime}\beta l-\rho^{\prime 2}\alpha^{2}}\right)\right]p_{1}q_{1}}\right.\right.\nonumber\\
&&\left.\left.\times e^{-i\left[\frac{(\rho-\rho^{\prime})\alpha}{2}-\frac{l\gamma\beta(1-m)}{\alpha}\left(\frac{\tau\sigma}{\rho}-\frac{\tau^{\prime}\sigma^{\prime}}{\rho^{\prime}}\right)\right]p_{2}q_{2}-i\left(m-\frac{1}{2}\right)\beta(\sigma-\sigma^{\prime})p_{1}p_{2}}\right.\right.\nonumber\\
&&\left.\left.\times e^{i\left[\frac{(\tau-\tau^{\prime})\gamma}{2}-\gamma(1-l)\left\{\frac{(\tau^{2}\gamma\sigma\beta l-\tau^{2}\gamma\sigma\beta lm-\tau\rho^{2}\alpha^{2})}{\tau\gamma\sigma\beta l-\rho^{2}\alpha^{2}}-\frac{(\tau^{\prime 2}\gamma\sigma^{\prime}\beta l-\tau^{\prime 2}\gamma\sigma^{\prime}\beta lm-\tau^{\prime 2}\rho^{\prime 2}\alpha^{2})}{\tau^{\prime}\gamma\sigma^{\prime}\beta l-\rho^{\prime 2}\alpha^{2}}\right\}\right]}\right.\right.\nonumber\\
&&\scriptscriptstyle\left.\left.\times\chi_{\rho,\sigma,\tau}(r_{1},r_{2})\overline{\lambda_{\rho,\sigma,\tau}\left(r_{1}-\frac{[1-m]\sigma\beta}{\rho\alpha}p_{2}+\frac{[\tau\gamma\sigma\beta(l+m-lm)-\rho^{2}\alpha^{2}]}{\tau\gamma\sigma\beta l-\rho^{2}\alpha^{2}}q_{1}, r_{2}+\frac{m\sigma\beta}{\rho\alpha}p_{1}-\frac{[\tau\gamma\sigma\beta l(1-m)-\rho^{2}\alpha^{2}]}{\rho^{2}\alpha^{2}}q_{2}\right)}\right.\right.\nonumber\\
&&\scriptscriptstyle\left.\left.\times\overline{\chi_{\rho^{\prime},\sigma^{\prime},\tau^{\prime}}(r^{\prime}_{1},r^{\prime}_{2})}\lambda_{\rho^{\prime},\sigma^{\prime},\tau^{\prime}}\left(r^{\prime}_{1}-\frac{[1-m]\sigma^{\prime}\beta}{\rho^{\prime}\alpha}p_{2}+\frac{\tau^{\prime}\gamma\sigma^{\prime}\beta(l+m-lm)-\rho^{\prime 2}\alpha^{2}}{\tau^{\prime}\gamma\sigma^{\prime}\beta l-\rho^{\prime 2}\alpha^{2}}q_{1}, r^{\prime}_{2}+\frac{m\sigma^{\prime}\beta}{\rho^{\prime}\alpha}p_{1}-\frac{[\tau^{\prime}\gamma\sigma^{\prime}\beta l(1-m)-\rho^{\prime 2}\alpha^{2}]}{\rho^{\prime 2}\alpha^{2}}q_{2}\right)\right.\right.\nonumber\\
&&\scriptstyle\left.\left.\times dr_{1}dr_{2}dr^{\prime}_{1}dr^{\prime}_{2}\right\} d\theta d\phi d\psi dq_{1} dq_{2} dp_{1} dp_{2}\right]\kappa(\rho,\sigma,\tau)\kappa(\rho^{\prime},\sigma^{\prime},\tau^{\prime})d\rho d\sigma d\tau d\rho^{\prime}d\sigma^{\prime} d\tau^{\prime}\nonumber\\
=\lefteqn{\textstyle\frac{(2\pi)^{3}}{N^2}\int_{(\rho,\sigma,\tau)}\left[\int_{\mathbb{R}^{4}}\left\{\int_{(r_{1},r_{2})\in\mathbb{R}^{2}}\int_{(r_{1}^{\prime},r_{2}^{\prime})\in\mathbb{R}^{2}}e^{i\alpha\rho p_{1}(r_{1}-r_{1}^{\prime})+i\alpha\rho p_{2}(r_{2}-r_{2}^{\prime})}\right.\right.}\nonumber\\
&&\scriptstyle\times\left.\left. e^{-\frac{i\rho^{2}\alpha^{2}\gamma\tau(1-l)q_{1}}{\tau\gamma\sigma\beta l-\rho^{2}\alpha^{2}}(r_{2}-r_{2}^{\prime})-il\gamma\tau q_{2}(r_{1}-r_{1}^{\prime})}\chi_{\rho,\sigma,\tau}(r_{1},r_{2})\right.\right.\nonumber\\
&&\scriptscriptstyle\left.\left.\times\overline{\lambda_{\rho,\sigma,\tau}\left(r_{1}-\frac{[1-m]\sigma\beta}{\rho\alpha}p_{2}+\frac{[\tau\gamma\sigma\beta(l+m-lm)-\rho^{2}\alpha^{2}]}{\tau\gamma\sigma\beta l-\rho^{2}\alpha^{2}}q_{1}, r_{2}+\frac{m\sigma\beta}{\rho\alpha}p_{1}-\frac{[\tau\gamma\sigma\beta l(1-m)-\rho^{2}\alpha^{2}]}{\rho^{2}\alpha^{2}}q_{2}\right)}\right.\right.\nonumber\\
&&\scriptscriptstyle\left.\left.\times\overline{\chi_{\rho,\sigma,\tau}(r^{\prime}_{1},r^{\prime}_{2})}\lambda_{\rho,\sigma,\tau}\left(r^{\prime}_{1}-\frac{[1-m]\sigma\beta}{\rho\alpha}p_{2}+\frac{\tau\gamma\sigma\beta(l+m-lm)-\rho^{\prime 2}\alpha^{2}}{\tau\gamma\sigma\beta l-\rho^{2}\alpha^{2}}q_{1}, r^{\prime}_{2}+\frac{m\sigma\beta}{\rho\alpha}p_{1}-\frac{[\tau\gamma\sigma\beta l(1-m)-\rho^{2}\alpha^{2}]}{\rho^{2}\alpha^{2}}q_{2}\right)\right.\right.\nonumber\\
&&\scriptstyle\left.\left.\times dr_{1}dr_{2}dr^{\prime}_{1}dr^{\prime}_{2}\right\}dq_{1} dq_{2} dp_{1} dp_{2}\right][\kappa(\rho,\sigma,\tau)]^{2}d\rho d\sigma d\tau\nonumber\\
=\lefteqn{\frac{(2\pi)^{3}}{N^2}\int_{(\rho,\sigma,\tau)}\left[\int_{\mathbb{R}^{4}}\left\{\int_{(r_{1},r_{2})\in\mathbb{R}^{2}}\int_{(r_{1}^{\prime},r_{2}^{\prime})\in\mathbb{R}^{2}}e^{-i\rho\alpha\left[\frac{\rho\alpha(\tau\gamma\sigma\beta-\rho^{2}\alpha^{2})q_{1}}{\sigma\beta(\tau\gamma\sigma\beta l-\rho^{2}\alpha^{2})(1-m)}+\frac{\Pi_{2}}{1-m}\right](r_{2}-r_{2}^{\prime})}\right.\right.}\nonumber\\
&&\left.\left.\times e^{-i\rho\alpha\left[\frac{(\tau\gamma\sigma\beta l-\rho^{2}\alpha^{2})q_{2}}{m\rho\alpha\sigma\beta}+\frac{\Pi_{1}}{m}\right](r_{1}-r_{1}^{\prime})}\chi_{\rho,\sigma,\tau}(r_{1},r_{2})\overline{\lambda_{\rho,\sigma,\tau}\left(r_{1}-\frac{\sigma\beta}{\rho\alpha}\Pi_{2},r_{2}+\frac{\sigma\beta}{\rho\alpha}\Pi_{1}\right)}\right.\right.\nonumber\\
&&\left.\left.\times\overline{\chi_{\rho,\sigma,\tau}(r_{1}^{\prime},r_{2}^{\prime})}\lambda_{\rho,\sigma,\tau}\left(r_{1}^{\prime}-\frac{\sigma\beta}{\rho\alpha}\Pi_{2},r_{2}^{\prime}+\frac{\sigma\beta}{\rho\alpha}\Pi_{1}\right)dr_{1}dr_{2}dr_{1}^{\prime}dr_{2}^{\prime}\right\}\right.\nonumber\\
&&\left.\times\frac{1}{|m(1-m)|}d\Pi_{1}d\Pi_{2}dq_{1}dq_{2}\right][\kappa(\rho,\sigma,\tau)]^{2}d\rho d\sigma d\tau, 
\end{eqnarray}
where we have made the following changes of variables:
\begin{equation}
\begin{split}
\frac{m\sigma\beta}{\rho\alpha}p_{1}-\frac{[\tau\gamma\sigma\beta l(1-m)-\rho^{2}\alpha^{2}]}{\rho^{2}\alpha^{2}}q_{2}&=\frac{\sigma\beta}{\rho\alpha}\Pi_{1}\\
\frac{(1-m)\sigma\beta}{\rho\alpha}p_{2}-\frac{[\tau\gamma\sigma\beta(l+m-lm)-\rho^{2}\alpha^{2}]}{\tau\gamma\sigma\beta l-\rho^{2}\alpha^{2}}q_{1}&=\frac{\sigma\beta}{\rho\alpha}\Pi_{2}.
\end{split}
\end{equation}
Now (\ref{eq:planchrl-dervtn}) reduces to
\begin{eqnarray}\label{Plancrl-dervtn-continued}
\lefteqn{\frac{(2\pi)^{5}}{N^2}\int_{(\rho,\sigma,\tau)}\frac{|\sigma\beta\gamma(\tau\gamma\sigma\beta l-\rho^{2}\alpha^{2})(1-m)|}{\rho^{2}\alpha^{2}|\tau\gamma\sigma\beta-\rho^{2}\alpha^{2}|}\times\frac{|m\sigma\beta\gamma|}{|\tau\gamma\sigma\beta l-\rho^{2}\alpha^{2}|}\times\frac{1}{|m(1-m)|}}\nonumber\\
&&\times\left[\int_{(\Pi_{1},\Pi_{2})\in\mathbb{R}^{2}}\left\{\int_{(r_{1},r_{2})\in\mathbb{R}^{2}}\chi_{\rho,\sigma,\tau}(r_{1},r_{2})\overline{\lambda_{\rho,\sigma,\tau}\left(r_{1}-\frac{\sigma\beta}{\rho\alpha}\Pi_{2},r_{2}+\frac{\sigma\beta}{\rho\alpha}\Pi_{1}\right)}\right.\right.\nonumber\\
&&\left.\left.\times\overline{\chi_{\rho,\sigma,\tau}(r_{1},r_{2})}\lambda_{\rho,\sigma,\tau}\left(r_{1}-\frac{\sigma\beta}{\rho\alpha}\Pi_{2},r_{2}+\frac{\sigma\beta}{\rho\alpha}\Pi_{1}\right)dr_{1}dr_{2}\right\}d\Pi_{1}d\Pi_{2}\right][\kappa(\rho,\sigma,\tau)]^{2}d\rho d\sigma d\tau\nonumber\\
=\lefteqn{\frac{(2\pi)^{5}}{N^2}\int_{(\rho,\sigma,\tau)}\frac{\sigma^{2}\beta^{2}\gamma^{2}}{\rho^{2}\alpha^{2}|\tau\gamma\sigma\beta-\rho^{2}\alpha^{2}|}\left[\int_{(r_{1},r_{2})\in\mathbb{R}^{2}}\chi_{\rho,\sigma,\tau}(r_{1},r_{2})\overline{\chi_{\rho,\sigma,\tau}(r_{1},r_{2})}\right.}\nonumber\\
&&\left.\times\left\{\int_{\left(\widetilde{\Pi}_{2},\widetilde{\Pi}_{1}\right)\in\mathbb{R}^{2}}\overline{\lambda_{\rho,\sigma,\tau}(\widetilde{\Pi}_{2},\widetilde{\Pi}_{1})}\lambda_{\rho,\sigma,\tau}(\widetilde{\Pi}_{2},\widetilde{\Pi}_{1})\frac{\rho^{2}\alpha^{4}}{\sigma^{2}\beta^{2}\gamma^{2}}d\widetilde{\Pi}_{1}d\widetilde{\Pi}_{2}\right\}dr_{1}dr_{2}\right]\nonumber\\
&&\times[\kappa(\rho,\sigma,\tau)]^{2}d\rho d\sigma d\tau,
\end{eqnarray}
where, in (\ref{Plancrl-dervtn-continued}), we have introduced the following change of variables:
\begin{equation}\label{second-change-vrble}
\begin{split}
r_{1}-\frac{\sigma\beta}{\rho\alpha}\Pi_{2}&=\frac{\alpha}{\gamma}\widetilde{\Pi}_{2}\\
r_{2}-\frac{\sigma\beta}{\rho\alpha}\Pi_{1}&=\frac{\alpha}{\gamma}\widetilde{\Pi}_{1},
\end{split}
\end{equation}
so that (\ref{Plancrl-dervtn-continued}) now takes the following simple form:
\begin{equation}\label{final-exprssn}
\frac{(2\pi)^{5}}{N^2}\int_{(\rho,\sigma,\tau)}\frac{\alpha^{2}}{|\tau\gamma\sigma\beta-\rho^{2}\alpha^{2}|}||\chi_{\rho,\sigma,\tau}||^{2}||\lambda_{\rho,\sigma,\tau}||^{2}[\kappa(\rho,\sigma,\tau)]^{2}d\rho d\sigma d\tau.
\end{equation}

Now comparing (\ref{final-exprssn}) with the right side of (\ref{eq:ortho-con-presnt-secnario}) then yields
\begin{equation}
\begin{split}
N&=(2\pi)^{\frac{5}{2}}\\
\kappa(\rho,\sigma,\tau)&=\frac{|\tau\gamma\sigma\beta-\rho^{2}\alpha^{2}|}{\alpha^{2}},
\end{split}
\end{equation}
proving the proposition.
\end{proof}

\begin{remark}
A few remarks on the statement of the proposition \ref{prop:Plancherel} and its proof are in order. In this paper, we deal with an arbitrary member of each family of equivalence classes of the sector $\mathfrak{N}$ of the unitary dual $\gd$ labeled by the ordered pair $(l,m)$. The underlying proposition states that for any choice of the representative (determined by $(l,m)$) from each equivalence class restricted to the sector $\mathfrak{N}$ of $\gd$, we obtain the given Plancherel measure.  It is to be noted that during the proof, one has to exercise caution while introducing change of variables (see (\ref{second-change-vrble})) where on the left side one has quantities with dimension of length. So, to bring the dimension of momentum , i.e. the one associated with the $\widetilde{\Pi}_{i}$'s, for $i=1,2$ on the right side of (\ref{second-change-vrble}), to that of length, we had to insert a factor of $\frac{\alpha}{\gamma}$ which indeed has the dimension of Length/Momentum. These careful change of variables yield the desired Plancherel measure.
\end{remark}

\subsection{Construction of Wigner function}
\label{subsec:wigner-func-constrctn}

Recall from \cite{wigfuncpaper} that the Lebesgue measure $dX^{*}$ on the dual Lie algebra $\G^{*}$ decomposes as
\begin{equation}
   dX^{*} = s_{\rho, \sigma, \tau} (X^*_{\rho, \sigma, \tau }) \; d\nu_{\g}(\rho,\sigma,\tau)\; d\Omega_{\rho, \sigma, \tau}(X^*_{\rho, \sigma, \tau }), \qquad X^*_{\rho, \sigma, \tau } \in \mathcal O^{\rho, \sigma, \tau}
\label{dual-meas-split}
\end{equation}
where $s_{\rho, \sigma, \tau}$ is  a positive density, $d\Omega_{\rho, \sigma, \tau}$ is the canonical invariant measure on the coadjoint orbit $\mathcal O^{\rho, \sigma, \tau}$ (in this case, just the Lebesgue measure on $\mathbb R^4$) and  $d\nu_{\g}$ is the Plancherel measure (\ref{eq:Plancherel-prop}). By $X^{*}$ and $X_{\rho,\sigma,\tau}^{*}$, we denote a generic point in $\G^{*}$ and in the 4-dimensional coadjoint orbit $\mathcal{O}^{\rho,\sigma,\tau}_{4}\in\G^{*}$, respectively. It then immediately follows that the strictly positive density $s_{\rho,\sigma,\tau}(X^{*}_{\rho,\sigma,\tau})$ is given by
\begin{equation}\label{eq:density}
s_{\rho,\sigma,\tau}(X^{*}_{\rho,\sigma,\tau})=\frac{\alpha^{2}}{\vert\rho^{2}\alpha^{2}-\gamma\beta\sigma\tau\vert},
\end{equation}
which upon substitution in the definition (see p. 6 of \cite{wigfuncpaper}) of the NCQM Wigner function, i.e.
\begin{eqnarray}
W(A;\; X^*_{\rho, \sigma, \tau })&=& \frac {[s_{\rho, \sigma, \tau} (X^*_{\rho, \sigma, \tau })]^{\frac 12}}{(2\pi)^\frac 72}\;\int_{\G}  e^{-i\langle X^*_{\rho, \sigma, \tau}; X\rangle}\nonumber\\
   &\times& \left[  \int_{\mathfrak N} \text{Tr}[ U^{\omega,\nu,\mu}_{l,m}(e^{-X}) A(\omega, \nu, \mu) C^{-1}_{\omega , \nu , \mu}]\; d\nu_{\g}(\omega, \nu, \mu)\right]dX, \;\;
\label{wig-fcn-def}
\end{eqnarray}
leads us to the following expression:
\begin{eqnarray}
 W(A;\; X^*_{\rho, \sigma, \tau }) &=&   \frac {\vert\alpha\vert}{(2\pi)^6 \;
 \vert\rho^{2}\alpha^{2}-\gamma\beta\sigma\tau\vert^{\frac 12}} \;\int_{\G}  e^{-i\langle X^*_{\rho, \sigma, \tau}; X\rangle}\nonumber\\
   &\times& \left[  \int_{\mathfrak N} \text{Tr}[ U^{\omega,\nu,\mu}_{l,m}(e^{-X}) A(\omega, \nu, \mu) ]\; d\nu_{\g}(\omega, \nu, \mu)\right]dX. \;\;
\label{wig-fcn-def2}
\end{eqnarray}

It has been argued in \cite{plancherel,Fuhrbook} that for nilpotent Lie group things become more tractable in the sense that the underlying Wigner functions are decomposable, meaning that the Wigner function only picks up the contribution of the Hilbert space of Hilbert-Schmidt operators associated with the underlying coadjoint orbit for the measurable operator fields $(\rho, \sigma, \tau ) \longmapsto A(\rho, \sigma, \tau ) \in \mathcal{B}_{2} (\rho, \sigma, \tau)$, so that
\begin{equation}
  W(A;\; X^*_{\rho, \sigma, \tau }) = [W_{\rho, \sigma, \tau} A(\rho, \sigma, \tau )] ( X^{*}_{\rho, \sigma, \tau }) := W (A(\rho, \sigma, \tau ); \; X^{*}_{\rho, \sigma, \tau }; \; \rho, \sigma, \tau),
\label{eq:decomposition}
\end{equation}
We are now in a state to provide the main theorem of the paper which is as follows
\begin{Theo}
\label{thm:construction-wigner-func}
The Wigner function for NCQM in 2-dimensions restricted to the 4-dimensional coadjoint orbit $\mathcal{O}^{k_{1},k_{2},k_{3}}_{4}$ for nonzero $k_{i}$'s satisfying $k_{1}^{2}\alpha^{2}-k_{2}k_{3}\gamma\beta\neq 0$ due to the 2-parameter family of UIRs of $\g$ (see (\ref{eq:generic-rep})) labeled by $(l,m)$ can be computed as 
\begin{eqnarray}
\lefteqn{ W^{l,m}(|\chi_{k_1,k_2,k_3}\rangle\langle\lambda_{k_1,k_2,k_3}|;\;k_{1}^{*},k_{2}^{*},k_{3}^{*},k_{4}^{*};\;k_1,k_2,k_3)}\nonumber\\
&&=\textstyle\frac{\vert\alpha\vert}{2\pi\vert k_{1}^{2}\alpha^{2}-k_{2}k_{3}\beta\gamma\vert^{\frac{1}{2}}}\int_{\mathbb{R}^{2}}e^{i\alpha\left[\frac{k_{1}\alpha k_{3}\gamma(1-l)k_{2}^{*}+(k_{1}^{2}\alpha^{2}-k_{2}\beta k_{3}\gamma l)k_{3}^{*}}{k_{2}\beta k_{3}\gamma-k_{1}^{2}\alpha^{2}}\right]\widetilde{q}_{1}+i\alpha\left[\frac{k_{1}\alpha k_{3}\gamma lk_{1}^{*}-k_{1}^{2}\alpha^{2}k_{4}^{*}}{k_{1}^{2}\alpha^{2}-k_{2}\beta k_{3}\gamma l}\right]\widetilde{q}_{2}}\nonumber\\
&&\scriptscriptstyle\times\overline{\lambda_{k_{1},k_{2},k_{3}}\left(\frac{\widetilde{q}_{1}}{2}+\frac{[k_{3}\gamma k_{2}\beta l(1-m)-k_{1}^{2}\alpha^{2}]k_{1}^{*}+mk_{1}k_{2}\alpha\beta k_{4}^{*}}{k_{1}(k_{1}^{2}\alpha^{2}-k_{2}\beta k_{3}\gamma l)},\frac{\widetilde{q}_{2}}{2}+\frac{k_{1}\alpha[k_{2}\beta k_{3}\gamma(l+m-lm)-k_{1}^{2}\alpha^{2}]k_{2}^{*}+(1-m)k_{2}\beta(k_{3}\gamma k_{2}\beta l-k_{1}^{2}\alpha^{2})k_{3}^{*}}{k_{1}^{2}\alpha(k_{1}^{2}\alpha^{2}-k_{2}\beta k_{3}\gamma)}\right)}\nonumber\\
&&\scriptstyle\times\chi_{k_{1},k_{2},k_{3}}\left(\frac{-\widetilde{q}_{1}}{2}+\frac{[k_{3}\gamma k_{2}\beta l(1-m)-k_{1}^{2}\alpha^{2}]k_{1}^{*}+mk_{1}k_{2}\alpha\beta k_{4}^{*}}{k_{1}(k_{1}^{2}\alpha^{2}-k_{2}\beta k_{3}\gamma l)},\frac{-\widetilde{q}_{2}}{2}+\frac{k_{1}\alpha[k_{2}\beta k_{3}\gamma(l+m-lm)-k_{1}^{2}\alpha^{2}]k_{2}^{*}+(1-m)k_{2}\beta(k_{3}\gamma k_{2}\beta l-k_{1}^{2}\alpha^{2})k_{3}^{*}}{k_{1}^{2}\alpha(k_{1}^{2}\alpha^{2}-k_{2}\beta k_{3}\gamma)}\right)\nonumber\\
&&\scriptstyle\times d\widetilde{q}_{1}d\widetilde{q}_{2},
\label{NC-Wig-func-exct-exprssn}
\end{eqnarray}
where the Hilbert-Schmidt operator $|\chi_{k_1,k_2,k_3}\rangle\langle\lambda_{k_1,k_2,k_3}|\in\mathcal{B}_{2}(k_1, k_2, k_3) = \mathcal{B}_{2}(L^{2}(\mathbb{R}^{2},dr_{1}dr_{2})$ with $L^{2}(\mathbb{R}^{2},dr_{1}dr_{2})$ being the representation space  of the UIRs of $\g$ given by (\ref{eq:generic-rep}).
\end{Theo}

\begin{proof}
Choose a generic element $g$ of $\g$ to be $(-\theta,-\phi,-\psi,-\bq,\bp)$ so that inverse group element $g^{-1}$ is given by $(\theta,\phi,\psi,\bq,-\bp)$. Now, using the definition given in (\ref{wig-fcn-def2}), the Wigner function of $\g$ restricted to the 4-dimensional coadjoint orbits $\mathcal{O}^{k_{1},k_{2},k_{3}}_{4}$ reads
\begin{align}
\MoveEqLeft W(|\hat{\chi}_{k_1,k_2,k_3}\rangle\langle\hat{\lambda}_{k_1,k_2,k_3}|;k_{1}^{*},k_{2}^{*},k_{3}^{*},k_{4}^{*};k_1,k_2,k_3)\nonumber\\
&=\frac{\vert\alpha\vert}{(2\pi)^{6}\vert k_{1}^{2}\alpha^{2}-k_{2}k_{3}\gamma\beta\vert^{\frac{1}{2}}}\int_{\mathbb{R}^{7}}e^{-i\alpha(k_{1}^{*}p_{1}+k_{2}^{*}p_{2}-k_{3}^{*}q_{1}-k_{4}^{*}q_{2})}e^{-i(-k_{1}\theta-k_{2}\phi-k_{3}\psi)}\nonumber\\
&\scriptstyle\times\left[\int_{(\omega,\nu,\mu)}\int_{(r_{1},r_{2})\in\mathbb{R}^{2}}e^{-i\omega\theta-i\nu\phi-i\mu\psi}e^{-i\omega\alpha p_{1}r_{1}-i\omega\alpha p_{2}r_{2}+\frac{i\omega^{2}\alpha^{2}\mu\nu(1-l)}{\mu\gamma\nu\beta l-\omega^{2}\alpha^{2}}q_{1}r_{2}+il\mu\gamma q_{2}r_{1}-i\left[\frac{\omega\alpha}{2}+\frac{\omega\alpha\mu\gamma\nu\beta m(1-l)}{\mu\gamma\nu\beta l-\omega^{2}\alpha^{2}}\right]p_{1}q_{1}}\right.\nonumber\\
&\scriptstyle\times\left.e^{-i\left[\frac{\omega\alpha}{2}-\frac{l\mu\gamma\nu\beta(1-m)}{\omega\alpha}\right]p_{2}q_{2}+i\left(m-\frac{1}{2}\right)\nu\beta p_{1}p_{2}+i\left[\frac{\mu\nu}{2}-\frac{\mu\gamma(1-l)(\mu\gamma\nu\beta l-\mu\gamma\nu\beta lm-\omega^{2}\alpha^{2})}{\mu\gamma\nu\beta l-\omega^{2}\alpha^{2}}\right]q_{1}q_{2}}\overline{\lambda_{\omega,\nu,\mu}(r_{1},r_{2})}\right.\nonumber\\
&\scriptstyle\left.\times\chi_{\omega,\nu,\mu}\left(r_{1}+\frac{[1-m]\nu\beta}{\omega\alpha}p_{2}+\frac{\mu\gamma\nu\beta(l+m-lm)-\omega^{2}\alpha^{2}}{\mu\gamma\nu\beta l-\omega^{2}\alpha^{2}}q_{1},r_{2}-\frac{m\nu\beta}{\omega\alpha}p_{1}-\frac{[\mu\gamma\nu\beta l(1-m)-\omega^{2}\alpha^{2}]}{\omega^{2}\alpha^{2}}q_{2}\right)\frac{|\omega^{2}\alpha^{2}-\gamma\beta\sigma\tau|}{\alpha^{2}}dr_{1}dr_{2}d\omega d\nu d\mu\right]\nonumber\\
&\scriptstyle\times d\theta d\phi d\psi dq_{1}dq_{2}dp_{1}dp_{2}\nonumber\\
&=\scriptstyle\frac{{|k_{1}^{2}\alpha^{2}-k_{2}k_{3}\gamma\beta|}^{\frac{1}{2}}}{(2\pi^{3}|\alpha|)}\int_{\mathbb{R}^{4}}e^{-i\alpha k_{1}^{*}p_{1}-i\alpha k_{2}^{*}p_{2}+i\alpha k_{3}^{*}q_{1}+i\alpha k_{4}^{*}q_{2}-i\left[\frac{k_{1}\alpha}{2}+\frac{k_{1}\alpha k_{3}\gamma k_{2}\beta m(1-l)}{k_{3}\gamma k_{2}\beta l-k_{1}^{2}\alpha^{2}}\right]p_{1}q_{1}-i\left[\frac{k_{1}\alpha}{2}-\frac{lk_{3}\gamma k_{2}\beta(1-m)}{k_{1}\alpha}\right]p_{2}q_{2}}\nonumber\\
&\scriptstyle\times e^{i\left(m-\frac{1}{2}\right)k_{2}\beta p_{1}p_{2}+i\left[\frac{k_{3}\gamma}{2}-\frac{k_{3}\gamma(1-l)(k_{3}\gamma k_{2}\beta l-k_{3}\gamma k_{2}\beta lm-k_{1}^{2}\alpha^{2})}{k_{3}\gamma k_{2}\beta l-k_{1}^{2}\alpha^{2}}\right]q_{1}q_{2}}\left[\int_{(r_{1},r_{2})}e^{-i\alpha k_{1}p_{1}r_{1}-i\alpha k_{1}p_{2}r_{2}+\frac{ik_{1}^{2}\alpha^{2}k_{3}\gamma(1-l)}{k_{3}\gamma k_{2}\beta l-k_{1}^{2}\alpha^{2}}q_{1}r_{2}}\right.\nonumber\\
&\scriptscriptstyle\left.\times e^{ilk_{3}\gamma q_{2}r_{1}}\overline{\lambda_{k_{1},k_{2},k_{3}}(r_{1},r_{2})}\chi_{k_{1},k_{2},k_{3}}\left(r_{1}+\frac{[1-m]k_{2}\beta}{k_{1}\alpha}p_{2}+\frac{k_{3}\gamma k_{2}\beta(l+m-lm)-k_{1}^{2}\alpha^{2}}{k_{3}\gamma k_{2}\beta l-k_{1}^{2}\alpha^{2}}q_{1},r_{2}-\frac{mk_{2}\beta}{k_{1}\alpha}p_{1}-\frac{[k_{3}\gamma k_{2}\beta l(1-m)-k_{1}^{2}\alpha^{2}]}{k_{1}^{2}\alpha^{2}}q_{2}\right)\right.\nonumber\\
&\scriptstyle\left.\times dr_{1}dr_{2}\right]dq_{1}dq_{2}dp_{1}dp_{2}\nonumber\\
&=\scriptstyle\frac{{|k_{1}^{2}\alpha^{2}-k_{2}k_{3}\gamma\beta|}^{\frac{1}{2}}}{(2\pi^{3}|\alpha|)}\int_{\mathbb{R}^{4}}e^{-i\alpha k_{1}^{*}\widetilde{p}_{1}-i\alpha k_{2}^{*}\widetilde{p}_{2}+\frac{i\alpha k_{3}^{*}(k_{1}^{2}\alpha^{2}-k_{3}\gamma k_{2}\beta l)}{[k_{3}\gamma k_{2}\beta(l+m-lm)-k_{1}^{2}\alpha^{2}]}\widetilde{q}_{1}-\frac{ik_{3}^{*}(1-m)k_{2}\beta(k_{3}\gamma k_{2}\beta l-k_{1}^{2}\alpha^{2})}{k_{1}[k_{3}\gamma k_{2}\beta(l+m-lm)-k_{1}^{2}\alpha^{2}]}\widetilde{p}_{2}}\nonumber\\
&\scriptstyle\times e^{\frac{ik_{4}^{*}k_{1}^{2}\alpha^{3}}{k_{3}\gamma k_{2}\beta l(1-m)-k_{1}^{2}\alpha^{2}}\widetilde{q}_{2}-\frac{imk_{4}^{*}k_{1}\alpha^{2}k_{2}\beta}{k_{3}\gamma k_{2}\beta l(1-m)-k_{1}^{2}\alpha^{2}}\widetilde{p}_{1}+ik_{1}\alpha\widetilde{q}_{1}\widetilde{p}_{1}+\frac{ik_{1}\alpha(k_{1}^{2}\alpha^{2}-k_{2}\beta k_{3}\gamma l)}{2[k_{2}\beta k_{3}\gamma(l+m-lm)-k_{1}^{2}\alpha^{2}]}\widetilde{q}_{1}\widetilde{p}_{1}}\nonumber\\
&\scriptstyle\times e^{\frac{imk_{2}\beta(k_{2}\beta k_{3}\gamma-k_{1}^{2}\alpha^{2})}{2[k_{2}\beta k_{3}\gamma(l+m-lm)-k_{1}^{2}\alpha^{2}]}\widetilde{p}_{1}\widetilde{p}_{2}-\frac{ilk_{3}\gamma k_{1}^{2}\alpha^{2}}{k_{3}\gamma k_{2}\beta l(1-m)-k_{1}^{2}\alpha^{2}}\widetilde{q}_{1}\widetilde{q}_{2}+\frac{ilmk_{1}\alpha k_{2}\beta k_{3}\gamma}{k_{3}\gamma k_{2}\beta l(1-m)-k_{1}^{2}\alpha^{2}}\widetilde{p}_{1}\widetilde{q}_{1}-\frac{ik_{1}^{3}\alpha^{3}}{2[k_{3}\gamma k_{2}\beta l(1-m)-k_{1}^{2}\alpha^{2}]}\widetilde{p}_{2}\widetilde{q}_{2}}\nonumber\\
&\scriptstyle\times e^{\frac{imk_{1}^{2}\alpha^{2}k_{2}\beta}{2[k_{3}\gamma k_{2}\beta l(1-m)-k_{1}^{2}\alpha^{2}]}\widetilde{p}_{1}\widetilde{p}_{2}-\frac{ik_{1}^{2}\alpha^{2}k_{3}\gamma(k_{1}^{2}\alpha^{2}-k_{3}\gamma k_{2}\beta l)}{2[k_{3}\gamma k_{2}\beta(l+m-lm)-k_{1}^{2}\alpha^{2}][k_{3}\gamma k_{2}\beta l(1-m)-k_{1}^{2}\alpha^{2}]}\widetilde{q}_{1}\widetilde{q}_{2}}\nonumber\\
&\scriptstyle\times e^{\frac{imk_{1}\alpha k_{2}\beta k_{3}\gamma(k_{1}^{2}\alpha^{2}-k_{3}\gamma k_{2}\beta l)}{2[k_{3}\gamma k_{2}\beta(l+m-lm)-k_{1}^{2}\alpha^{2}][k_{3}\gamma k_{2}\beta l(1-m)-k_{1}^{2}\alpha^{2}]}\widetilde{p}_{1}\widetilde{q}_{1}+\frac{i(1-m)k_{1}\alpha k_{2}\beta k_{3}\gamma(k_{3}\gamma k_{2}\beta l-k_{1}^{2}\alpha^{2})}{2[k_{3}\gamma k_{2}\beta(l+m-lm)-k_{1}^{2}\alpha^{2}][k_{3}\gamma k_{2}\beta l(1-m)-k_{1}^{2}\alpha^{2}]}\widetilde{p}_{2}\widetilde{q}_{2}}\nonumber\\
&\scriptstyle\times {e^{-\frac{im(1-m)k_{2}^{2}\beta^{2}k_{3}\gamma(k_{3}\gamma k_{2}\beta l-k_{1}^{2}\alpha^{2})}{2[k_{3}\gamma k_{2}\beta(l+m-lm)-k_{1}^{2}\alpha^{2}][k_{3}\gamma k_{2}\beta l(1-m)-k_{1}^{2}\alpha^{2}]}\widetilde{p}_{1}\widetilde{p}_{2}}}\left[\int_{(r_{1},r_{2})}e^{-i\alpha k_{1}\widetilde{p}_{1}r_{1}-i\alpha k_{1}\widetilde{p}_{2}r_{2}-\frac{ik_{1}^{2}\alpha^{2}k_{3}\gamma(1-l)}{[k_{3}\gamma k_{2}\beta(l+m-lm)-k_{1}^{2}\alpha^{2}]}\widetilde{q}_{1}r_{2}}\right.\nonumber\\
&\scriptstyle\left.\times e^{-\frac{i(1-l)(1-m)k_{1}\alpha k_{2}\beta k_{3}\gamma}{[k_{3}\gamma k_{2}\beta(l+m-lm)-k_{1}^{2}\alpha^{2}]}\widetilde{p}_{2}r_{2}+\frac{ilk_{3}\gamma k_{1}^{2}\alpha^{2}}{[k_{3}\gamma k_{2}\beta l(1-m)-k_{1}^{2}\alpha^{2}]}\widetilde{q}_{2}r_{1}-\frac{ilmk_{1}\alpha k_{2}\beta k_{3}\gamma}{[k_{3}\gamma k_{2}\beta l(1-m)-k_{1}^{2}\alpha^{2}]}\widetilde{p}_{1}r_{1}}\overline{\lambda_{k_{1},k_{2},k_{3}}(r_{1},r_{2})}\right.\nonumber\\
&\scriptstyle\left.\times\chi_{k_{1},k_{2},k_{3}}(r_{1}-\widetilde{q}_{1},r_{2}-\widetilde{q}_{2})dr_{1}dr_{2}\right]\frac{k_{1}^{2}\alpha^{2}|k_{1}^{2}\alpha^{2}-k_{3}\gamma k_{2}\beta l|}{|k_{3}\gamma k_{2}\beta(l+m-lm)-k_{1}^{2}\alpha^{2}||k_{3}\gamma k_{2}\beta l(1-m)-k_{1}^{2}\alpha^{2}|}d\widetilde{q}_{1}d\widetilde{q}_{2}d\widetilde{p}_{1}d\widetilde{p}_{2},
\label{eq:NC-Wig-func-dervtn-first-prt}
\end{align}
where in the last line we have introduced the following change of variables:
\begin{equation}\label{change-of-var-proof-Wigfunc}
\begin{split}
\widetilde{q}_{1}&=-\frac{[k_{3}\gamma k_{2}\beta(l+m-lm)-k_{1}^{2}\alpha^{2}]}{k_{3}\gamma k_{2}\beta l-k_{1}^{2}\alpha^{2}}q_{1}-\frac{(1-m)k_{2}\beta}{k_{1}\alpha}p_{2},\\
\widetilde{q}_{2}&=\frac{k_{3}\gamma k_{2}\beta l(1-m)-k_{1}^{2}\alpha^{2}}{k_{1}^{2}\alpha^{2}}q_{2}+\frac{mk_{2}\beta}{k_{1}\alpha}p_{1},\\
\widetilde{p}_{1}&=p_{1},\\
\widetilde{p}_{2}&=p_{2}.
\end{split}
\end{equation}
Now (\ref{eq:NC-Wig-func-dervtn-first-prt}) can be manipulated to lead to the following expression
\begin{eqnarray}\label{Wigner-function-proof-second}
\lefteqn{\textstyle\frac{|\alpha|}{2\pi{|k_{1}^{2}\alpha^{2}-k_{2}\beta k_{3}\gamma|}^{\frac{1}{2}}}\int_{(\widetilde{q}_{1},\widetilde{q}_{2})\in\mathbb{R}^{2}}e^{\frac{i\alpha k_{3}^{*}(k_{1}^{2}\alpha^{2}-k_{3}\gamma k_{2}\beta l)}{[k_{3}\gamma k_{2}\beta(l+m-lm)-k_{1}^{2}\alpha^{2}]}\widetilde{q}_{1}+\frac{ik_{4}^{*}k_{1}^{2}\alpha^{3}}{[k_{3}\gamma k_{2}\beta l(1-m)-k_{1}^{2}\alpha^{2}]}\widetilde{q}_{2}-\frac{ik_{1}^{2}\alpha^{2}k_{3}\gamma l}{[k_{3}\gamma k_{2}\beta l(1-m)-k_{1}^{2}\alpha^{2}]}\widetilde{q}_{1}\widetilde{q}_{2}}}\nonumber\\
&&\scriptstyle\times e^{-\frac{ik_{3}\gamma k_{1}^{2}\alpha^{2}(k_{1}^{2}\alpha^{2}-k_{3}\gamma k_{2}\beta l)}{2[k_{3}\gamma k_{2}\beta(l+m-lm)-k_{1}^{2}\alpha^{2}][k_{3}\gamma k_{2}\beta l(1-m)-k_{1}^{2}\alpha^{2}]}\widetilde{q}_{1}\widetilde{q}_{2}}\left[\int_{(r_{1},r_{2})}e^{-\frac{ik_{1}^{2}\alpha^{2}k_{3}\gamma(1-l)}{[k_{3}\gamma k_{2}\beta(l+m-lm)-k_{1}^{2}\alpha^{2}]}\widetilde{q}_{1}r_{2}+\frac{ik_{1}^{2}\alpha^{2}k_{3}\gamma l}{[k_{2}\beta k_{3}\gamma l(1-m)-k_{1}^{2}\alpha^{2}]}\widetilde{q}_{2}r_{1}}\right.\nonumber\\
&&\scriptstyle\left.\times \delta\left(r_{1}-\frac{k_{1}^{*}[k_{3}\gamma k_{2}\beta l(1-m)-k_{1}^{2}\alpha^{2}]}{k_{1}(k_{1}^{2}\alpha^{2}-k_{2}\beta k_{3}\gamma l)}-\frac{mk_{4}^{*}\alpha k_{2}\beta}{[k_{1}^{2}\alpha^{2}-k_{2}\beta k_{3}\gamma l]}-\frac{\widetilde{q}_{1}}{2}\right)\delta\left(r_{2}-\frac{k_{2}^{*}[k_{2}\beta k_{3}\gamma(l+m-lm)-k_{1}^{2}\alpha^{2}]}{k_{1}(k_{1}^{2}\alpha^{2}-k_{2}\beta k_{3}\gamma)}-\frac{(1-m)k_{3}^{*}k_{2}\beta(k_{3}\gamma k_{2}\beta l-k_{1}^{2}\alpha^{2})}{k_{1}^{2}\alpha(k_{1}^{2}\alpha^{2}-k_{2}\beta k_{3}\gamma)}-\frac{\widetilde{q}_{2}}{2}\right)\right.\nonumber\\
&&\scriptstyle\left.\times\overline{\lambda_{k_{1},k_{2},k_{3}}(r_{1},r_{2})}\chi_{k_{1},k_{2},k_{3}}(r_{1}-\widetilde{q}_{1},r_{2}-\widetilde{q}_{2})dr_{1}dr_{2}\right]d\widetilde{q}_{1}d\widetilde{q}_{2}\nonumber\\
&&=\textstyle\frac{\vert\alpha\vert}{2\pi\vert k_{1}^{2}\alpha^{2}-k_{2}k_{3}\beta\gamma\vert^{\frac{1}{2}}}\int_{\mathbb{R}^{2}}e^{i\alpha\left[\frac{k_{1}\alpha k_{3}\gamma(1-l)k_{2}^{*}+(k_{1}^{2}\alpha^{2}-k_{2}\beta k_{3}\gamma l)k_{3}^{*}}{k_{2}\beta k_{3}\gamma-k_{1}^{2}\alpha^{2}}\right]\widetilde{q}_{1}+i\alpha\left[\frac{k_{1}\alpha k_{3}\gamma lk_{1}^{*}-k_{1}^{2}\alpha^{2}k_{4}^{*}}{k_{1}^{2}\alpha^{2}-k_{2}\beta k_{3}\gamma l}\right]\widetilde{q}_{2}}\nonumber\\
&&\scriptscriptstyle\times\overline{\lambda_{k_{1},k_{2},k_{3}}\left(\frac{\widetilde{q}_{1}}{2}+\frac{[k_{3}\gamma k_{2}\beta l(1-m)-k_{1}^{2}\alpha^{2}]k_{1}^{*}+mk_{1}k_{2}\alpha\beta k_{4}^{*}}{k_{1}(k_{1}^{2}\alpha^{2}-k_{2}\beta k_{3}\gamma l)},\frac{\widetilde{q}_{2}}{2}+\frac{k_{1}\alpha[k_{2}\beta k_{3}\gamma(l+m-lm)-k_{1}^{2}\alpha^{2}]k_{2}^{*}+(1-m)k_{2}\beta(k_{3}\gamma k_{2}\beta l-k_{1}^{2}\alpha^{2})k_{3}^{*}}{k_{1}^{2}\alpha(k_{1}^{2}\alpha^{2}-k_{2}\beta k_{3}\gamma)}\right)}\nonumber\\
&&\scriptstyle\times\chi_{k_{1},k_{2},k_{3}}\left(\frac{-\widetilde{q}_{1}}{2}+\frac{[k_{3}\gamma k_{2}\beta l(1-m)-k_{1}^{2}\alpha^{2}]k_{1}^{*}+mk_{1}k_{2}\alpha\beta k_{4}^{*}}{k_{1}(k_{1}^{2}\alpha^{2}-k_{2}\beta k_{3}\gamma l)},\frac{-\widetilde{q}_{2}}{2}+\frac{k_{1}\alpha[k_{2}\beta k_{3}\gamma(l+m-lm)-k_{1}^{2}\alpha^{2}]k_{2}^{*}+(1-m)k_{2}\beta(k_{3}\gamma k_{2}\beta l-k_{1}^{2}\alpha^{2})k_{3}^{*}}{k_{1}^{2}\alpha(k_{1}^{2}\alpha^{2}-k_{2}\beta k_{3}\gamma)}\right)\nonumber\\
&&\scriptstyle\times d\widetilde{q}_{1}d\widetilde{q}_{2}.
\end{eqnarray}
\end{proof}

\begin{remark}
\label{remark-wigfunc}
A few remarks on theorem \ref{thm:construction-wigner-func} are in order. First, we note that (\ref{NC-Wig-func-exct-exprssn}) can be rewritten as
\begin{align}
\MoveEqLeft W(A_{\rho,\sigma,\tau};\bq_{\mathrm{nc}},\bp_{\mathrm{nc}};k_1,k_2,k_3)\nonumber\\
&=\frac{\vert\alpha\vert}{2\pi\vert k_{1}^{2}\alpha^{2}-k_{2}k_{3}\beta\gamma\vert^{\frac{1}{2}}}\int_{\mathbb{R}^{2}}e^{-i\alpha p^{l}_{1,\mathrm{nc}}r_{1}-i\alpha p^{l}_{2,\mathrm{nc}}r_{2}}\;\overline{\lambda_{k_{1},k_{2},k_{3}}\left(\frac{1}{2}r_{1}-\frac{q^{l,m}_{1,\mathrm{nc}}}{k_{1}},\frac{1}{2}r_{2}-\frac{q^{l,m}_{2,\mathrm{nc}}}{k_{1}}\right)}\nonumber\\
&\times\chi_{k_{1},k_{2},k_{3}}\left(-\frac{1}{2}r_{1}-\frac{q^{l,m}_{1,\mathrm{nc}}}{k_{1}},-\frac{1}{2}r_{2}-\frac{q^{l,m}_{2,\mathrm{nc}}}{k_{1}}\right)dr_{1}dr_{2},
\label{reprmtrzd-exprssn-wig-func}
\end{align}
with the {\em ``noncommuting coordinates"} $\bq^{l,m}_{\mathrm{nc}}\equiv(q^{l,m}_{1,\mathrm{nc}},q^{l,m}_{2,\mathrm{nc}}),\bp^{l}_{\mathrm{nc}}\equiv(p^{l}_{1,\mathrm{nc}},p^{l}_{2,\mathrm{nc}})$, in terms of the phase space coordinates $k_{1}^{*}$, $k_{2}^{*}$, $k_{3}^{*}$ and $k_{4}^{*}$ can be read off immediately as
\begin{equation}
\label{int-vrbls-reprmtrzn}
\begin{aligned}
&q^{l,m}_{1,\mathrm{nc}}=\frac{[k_{3}\gamma k_{2}\beta l(1-m)-k_{1}^{2}\alpha^{2}]k_{1}^{*}+mk_{1}k_{2}\alpha\beta k_{4}^{*}}{(k_{2}\beta k_{3}\gamma l-k_{1}^{2}\alpha^{2})},\\
&q^{l,m}_{2,\mathrm{nc}}=\frac{k_{1}\alpha[k_{2}\beta k_{3}\gamma(l+m-lm)-k_{1}^{2}\alpha^{2}]k_{2}^{*}+(1-m)k_{2}\beta(k_{3}\gamma k_{2}\beta l-k_{1}^{2}\alpha^{2})k_{3}^{*}}{k_{1}\alpha(k_{2}\beta k_{3}\gamma-k_{1}^{2}\alpha^{2})},\\
&p^{l}_{1,\mathrm{nc}}=\frac{k_{1}\alpha k_{3}\gamma(1-l)k_{2}^{*}+(k_{1}^{2}\alpha^{2}-k_{2}\beta k_{3}\gamma l)k_{3}^{*}}{k_{1}^{2}\alpha^{2}-k_{2}\beta k_{3}\gamma},\\
&p^{l}_{2,\mathrm{nc}}=\frac{k_{1}^{2}\alpha^{2}k_{4}^{*}-k_{1}\alpha k_{3}\gamma lk_{1}^{*}}{k_{1}^{2}\alpha^{2}-k_{2}\beta k_{3}\gamma l}.
\end{aligned}
\end{equation}
Note the absence of the gauge parameter $m$ in the noncommuting momenta coordinates $\bp^{l}_{\mathrm{nc}}$ as given in (\ref{int-vrbls-reprmtrzn}). Now, $l=1$ and $m=0$ in (\ref{int-vrbls-reprmtrzn}) yields the noncommuting coordinates in the {\em Landau gauge} representation (see section \ref{sec:gauge-class} for details) as given below
\begin{equation}
\label{int-vrbls-reprmtrzn-landau gauge}
\begin{aligned}
q^{1,0}_{1,\mathrm{nc}}&=k_{1}^{*},\\
q^{1,0}_{2,\mathrm{nc}}&=\frac{k_{1}\alpha k_{2}^{*}+k_{2}\beta k_{3}^{*}}{k_{1}\alpha},\\
p^{1}_{1,\mathrm{nc}}&=k_{3}^{*},\\
p^{1}_{2,\mathrm{nc}}&=\frac{k_{1}^{2}\alpha^{2}k_{4}^{*}-k_{1}k_{3}\alpha\gamma k_{1}^{*}}{k_{1}^{2}\alpha^{2}-k_{2}k_{3}\beta\gamma}.
\end{aligned}
\end{equation}

For $l=1$ and $m=0$, by inserting the Fourier transform of $\lambda_{k_1,k_2,k_3},\chi_{k_1,k_2,k_3}\in L^{2}(\mathbb{R}^{2},dr_{1}dr_{2})$ in (\ref{reprmtrzd-exprssn-wig-func}), one obtains the NCQM Wigner function in the Landau gauge representation. Note that in \cite{wigfuncpaper}, Landau gauge representation of NCQM was taken to be the one corresponding to $l=m=1$ in (\ref{gauge-equiv-reps-algbr}). Here, we choose a unitary irreducible representation of $\g$ due to $l=1$ and $m=0$ in (\ref{eq:generic-rep}) that is unitarily equivalent to the one used in \cite{wigfuncpaper} and call it the Landau gauge representation. It is noteworthy that (2.17) of \cite{wigfuncpaper} can be obtained from (\ref{int-vrbls-reprmtrzn}) by plugging in $l=m=1$.

Recall from \cite{plethora} (see remark V.2) that one obtains the {\em symmetric gauge} representation of NCQM  by substituting $l=\frac{\rho\alpha(\rho\alpha-\sqrt{\rho^{2}\alpha^{2}-\tau\gamma\sigma\beta})}{\tau\gamma\sigma\beta}$ and $m=\frac{1}{2}$ in (\ref{eq:generic-rep}). But, there, the UIRs of $\G$ correspond to the coadjoint orbits associated with the nonzero fixed triple $(\rho,\sigma,\tau)$ satisfying $\rho^{2}\alpha^{2}-\tau\gamma\sigma\beta\neq 0$. In this paper, our study concerns generic coadjoint orbits determined by $\rho=k_1$, $\sigma=k_2$ and $\tau=k_3$ satisfying $k_{1}^{2}\alpha^{2}-k_{2}\beta k_{3}\gamma\neq 0$ so that the values of the 2-parameters $l$ and $m$ that determine the symmetric gauge condition are given by $l_{s}:=\frac{k_{1}\alpha(k_{1}\alpha-\sqrt{k_{1}^{2}\alpha^{2}-k_{2}\beta k_{3}\gamma})}{k_{2}\beta k_{3}\gamma}$ and $m=\frac{1}{2}$.  For the symmetric gauge representation of NCQM, the noncommuting coordinates (\ref{int-vrbls-reprmtrzn}) now read
\begin{equation}
\label{int-vrbls-reprmtrzn-symmetric-gauge}
\begin{split}
\textstyle q^{l_{s},\frac{1}{2}}_{1,\mathrm{nc}}&=\textstyle\frac{(k_{1}\alpha+\sqrt{k_{1}^{2}\alpha^{2}-k_{2}\beta k_{3}\gamma})k_{1}^{*}-k_{2}\beta k_{4}^{*}}{2\sqrt{k_{1}^{2}\alpha^{2}-k_{2}\beta k_{3}\gamma}},\\
\textstyle q^{l_{s},\frac{1}{2}}_{2,\mathrm{nc}}&=\textstyle\frac{(k_{1}^{2}\alpha^{2}-k_{2}\beta k_{3}\gamma+k_{1}\alpha\sqrt{k_{1}^{2}\alpha^{2}-k_{2}\beta k_{3}\gamma})k_{2}^{*}+(k_{2}\beta\sqrt{k_{1}^{2}\alpha^{2}-k_{2}\beta k_{3}\gamma})k_{3}^{*}}{2(k_{1}^{2}\alpha^{2}-k_{2}\beta k_{3}\gamma)},\\
\textstyle p^{l_{s}}_{1,\mathrm{nc}}&=\textstyle\frac{k_{1}\alpha(k_{2}\beta k_{3}\gamma-k_{1}^{2}\alpha^{2}+k_{1}\alpha\sqrt{k_{1}^{2}\alpha^{2}-k_{2}\beta k_{3}\gamma})k_{2}^{*}+(k_{1}\alpha k_{2}\beta\sqrt{k_{1}^{2}\alpha^{2}-k_{2}\beta k_{3}\gamma})k_{3}^{*}}{k_{2}\beta(k_{1}^{2}\alpha^{2}-k_{2}\beta k_{3}\gamma)},\\
\textstyle p^{l_{s}}_{2,\mathrm{nc}}&=\textstyle\frac{k_{1}\alpha k_{2}\beta k_{4}^{*}-k_{1}\alpha(k_{1}\alpha-\sqrt{k_{1}^{2}\alpha^{2}-k_{2}\beta k_{3}\gamma})k_{1}^{*}}{k_{2}\beta\sqrt{k_{1}^{2}\alpha^{2}-k_{2}\beta k_{3}\gamma}}.
\end{split}
\end{equation}

Substituting (\ref{int-vrbls-reprmtrzn-symmetric-gauge}) into (\ref{reprmtrzd-exprssn-wig-func}) then yields the expression of the NCQM Wigner function due to the symmetric gauge representation of $\g$.

\end{remark}

\section{Conclusion and future perspectives}

In the above, we have calculated the NC Wigner functions for the general gauge setting with the explicit two-gauge parameter $(l,m)$-dependence. The gauge parameter $l$ is associated to the gauge choice of the magnetic field and with specific settings can be shown to be reducible to either the Landau or symmetric gauge choices. It is also of interest to see how the spatial noncommutativity influences the computation explicitly. In the last two equations of \ref{gauge-equiv-reps-algbr-rarrngd}, the parameter $\sigma$ that signifies the spatial noncommutativity appears in the coefficient of the derivative operators. This in turn modifies the momenta commutator or the magnetic field expression. However in the limit of $\sigma\rightarrow 0$, one recovers the normal magnetic field in the momenta commutator. The modification of the magnetic field by the spatial noncommutativity seems to be in agreement with the work of Delduc et al. \cite{DelducEtAl2008}.

While the group-theoretic setting ensures the equivalence of representations under a unitarily equivalent class, further work can be done on the gauge equivalence of the Wigner functions under this two-parameter gauge setting by making the equivalence explicit. In particular, it is of interest to us to consider a new form of *-product equivalence in relation to the gauge equivalence. This will then complete the intertwining of the gauge symmetries available to the system with its quantum kinematical symmetries.

\end{document}